%% file: main.tex
\documentclass[a4paper,10pt,twoside,reqno,nonamelimits]{article}
\usepackage{a4wide}
\usepackage[colorlinks=true,urlcolor=blue,linkcolor=blue,citecolor=blue]{hyperref}
\usepackage{newcent}         

\usepackage{graphicx}
\usepackage{amsmath}
\usepackage{amssymb}
\usepackage{amsfonts}
\usepackage{amscd}
\usepackage{amsthm}
\usepackage{amsxtra}
\usepackage{mathrsfs}
\usepackage{nicefrac}
\usepackage{bbold}
\usepackage{graphicx}
\usepackage{xcolor}
\usepackage{enumitem}
\setlist[enumerate]{leftmargin=2em,itemindent=0em, labelindent=0pt,labelwidth=1.5em,labelsep=.5em, align=left, noitemsep}
\newlist{txtenum}{enumerate}{1}
\setlist[txtenum]{leftmargin=0em,itemindent=1.5em, labelindent=0pt,labelwidth=1em,labelsep=.5em, align=left}

\input{defsN.amsart}
\input{defs.common}

\newcommand{\mypar}{\par\medskip\noindent}
\newcommand{\myparbig}{\par\bigskip\noindent}

\DeclareMathOperator{\GF}{GF}
\renewcommand{\rho}{\varrho}

\begin{document}
\title{The (minimum) rank of typical fooling-set matrices}%
\author{Mozhgan Pourmoradnasseri, Dirk Oliver Theis\thanks{Supported by the Estonian Research Council, ETAG (\textit{Eesti Teadusagentuur}), through PUT Exploratory Grant \#620, and by the European Regional Development Fund through the Estonian Center of Excellence in Computer Science, EXCS.}\\[1ex]
  \small Institute of Computer Science {\tiny of the } University of Tartu\\
  \small \"Ulikooli 17, 51014 Tartu, Estonia\\
  \small \texttt{$\{$dotheis,mozhgan$\}$@ut.ee}%
}
%

\date{Mon Dec  5 21:37:55 EET 2016}

\maketitle
\begin{abstract}
  A fooling-set matrix has nonzero diagonal, but at least one in every pair of diagonally opposite entries is~0.  Dietzfelbinger et al.~'96 proved that the rank of such a matrix is at least~$\sqrt n$.  It is known that the bound is tight (up to a multiplicative constant).

  We ask for the \textit{typical} minimum rank of a fooling-set matrix: For a fooling-set zero-nonzero pattern chosen at random, is the minimum rank of a matrix with that zero-nonzero pattern over a field~$\FF$ closer to its lower bound $\sqrt{n}$ or to its upper bound~$n$?  We study random patterns with a given density~$p$, and prove an $\Omega(n)$ bound for the cases when
  \begin{enumerate}[label=(\alph{*}),nosep]
  \item $p$ tends to~$0$ quickly enough;
  \item $p$ tends to~$0$ slowly, and $\abs{\FF}=O(1)$;
  \item $p\in\lt]0,1\rt]$ is a constant.
  \end{enumerate}
  We have to leave open the case when $p\to 0$ slowly and $\FF$ is a large or infinite field (e.g., $\FF=\GF(2^n)$, $\FF=\RR$).
\end{abstract}

\section{Introduction}\label{sec:intro}
Let $f\colon X\times Y\to \{0,1\}$ be a function.  A \textit{fooling set} of size~$n$ is a family $(x_1,y_1),\dots,(x_n,y_n) \in X\times Y$ such that $f(x_i,y_i) = 1$ for all~$i$, and for $i\ne j$, at least one of $f(x_i,y_j)$ of $f(y_i,y_j)$ is~0.  Sizes of fooling sets are important lower bounds in Communication Complexity (see, e.g., \cite{Kushilevitz-Nisan:Book:97,Klauck-deWolf:foolQ:2013}) and the study of extended formulations (e.g., \cite{Fiorini-Kaibel-Pashkovich-Theis:CombLB:13,Beasley-Klauck-Lee-Theis:Dagstuhl:13}).

There is an \textsl{a priori} upper bound on the size of fooling sets due to Dietzfelbinger et al.~\cite{Dietzfelbinger-Hromkovic-Schnitger:96}, based on the rank of a matrix associated with~$f$.
Let~$\FF$ be an arbitrary field.  The following is a slight generalization of the result in~\cite{Dietzfelbinger-Hromkovic-Schnitger:96} (see the appendix for a proof).

\begin{lemma}\label{lem:dietzfelbinger-minrk}
  No fooling set in~$f$ is larger than the square of $\min_A \rk_\FF(A)$, where the minimum ranges\footnote{%
    This concept of \textit{minimum rank} differs from the definition used in the context of index coding \cite{Haviv-Langberg:minrk:2012,Golovnev-Regev-Weinstein:minrk:2016}.  It is closer to the minimum rank of a graph, but there the matrix~$A$ has to be symmetric while the diagonal entries are unconstrained. %
  } %
  over all $X\times Y$-matrices~$A$ over~$\FF$ with $A_{x,y} = 0$ iff $f(x,y)=0$.
\end{lemma}

It is known that, for fields~$\FF$ with nonzero characteristic, this upper bound is asymptotically attained~\cite{FriesenTheis13}, and for all fields, it is attained up to a multiplicative constant~\cite{Friesen-Hamed-Lee-Theis:fool:15}.  These results, however, require sophisticated constructions.  In this paper, we ask how useful that upper bound is for \textit{typical} functions~$f$.

Put differently, a \textit{fooling-set pattern of size~$n$} is a matrix~$R$ with entries in $\{0,1\} \subseteq \FF$ with $R_{k,k}=1$ for all~$k$ and $R_{k,\ell}R_{\ell,k}=0$ whenever $k\ne\ell$.  We say that a fooling-set pattern of size~$n$ has \textit{density} $p\in\lt]0,1\rt]$, if it has exactly $\lceil p\binom{n}{2} \rceil$ off-diagonal 1-entries.  So, the density is roughly the quotient $(\abs{R}-n)/\binom{n}{2}$, where $\abs{\cdot}$ denotes the Hamming weight, i.e., the number of nonzero entries.  The densest possible fooling-set pattern has $\binom{n}{2}$ off-diagonal ones (density $p=1$).

For any field~$\FF$ and $y\in \FF$, let
$\sigma(y) :=0$, if $y=0$, and $\sigma(y) :=1$, otherwise.
For a matrix (or vector, in case $n=1$) $M \in \FF^{m\times n}$,  define the \textit{zero-nonzero pattern of~$M$,} $\sigma(M)$, as the matrix in $\{0,1\}^{m\times n}$ which results from applying~$\sigma$ to every entry of~$M$.

\medskip\noindent%
This paper deals with the following question: \textsl{For a fooling-set pattern chosen at random, is the minimum rank of closer to its lower bound $\sqrt{n}$ or to its trivial upper bound~$n$?}  The question turns out to be surprisingly difficult.  We give partial results, but we must leave some cases open.
%
%
%
The distributions we study are the following:
\begin{description}
\item[$Q(n)$] denotes a fooling-set pattern drawn uniformly at random from all fooling-set patterns of size~$n$;
\item[$R(n,{p})$] denotes a fooling-set patterns drawn uniformly at random from all fooling-set patterns of size~$n$ with density~${p}$.
\end{description}
We allow that the density depends on the size of the matrix: ${p} = {p}(n)$.  From now on,  $Q=Q(n)$ and $R=R(n,{p})$ will denote these random fooling-set patterns.

Our first result is the following.  As customary, we use the terminology ``asymptotically almost surely, a.a.s.,'' to stand for ``with probability tending to~1 as~$n$ tends to infinity''.
\begin{theorem}\label{thm:main}
  \begin{enumerate}[label=(\alph*)]
  \item\label{thm:main:tiny-p} For every field $\FF$, if $p=O(1/n)$, then, a.a.s., the minimum rank of a matrix with zero-nonzero pattern $R(n,p)$ is $\Omega(1)$.
  \item\label{thm:main:small-p} Let $\FF$ be a finite field and $F := \abs{\FF}$.  (We allow~$F$ to grow with~$n$.)  If $100\max(1,\ln\ln F)/n \le p \le 1$, then the minimum rank of a matrix over~$\FF$ with zero-nonzero pattern $R(n,p)$ is
    \begin{equation*}
      \Omega\Bigl( \frac{\log(1/p)}{\log(1/p) + \log(F)}\;n \Bigr) = \Omega(n/\log(F)).
    \end{equation*}
  \item\label{thm:main:const-p} For every field $\FF$, if $p\in\lt]0,1\rt]$ is a constant, then the minimum rank of a matrix with zero-nonzero pattern $R(n,p)$ is $\Omega(1)$.  (The same is true for zero-nonzero pattern~$Q(n)$.)
  \end{enumerate}
\end{theorem}

Since the constant in the big-$\Omega$ in Thereom~\ref{thm:main}\ref{thm:main:const-p} tends to~$0$ with $p\to 0$,
the proof technique used for constant~$p$ does not work for $p=o(1)$; moreover, the bound in~\ref{thm:main:small-p} does not give an $\Omega(n)$ lower bound for infinite fields, or for large finite fields, e.g., $\GF(2^n)$.  We conjecture that the bound is still true (see Lemma~\ref{lem:indep-no} for a lower bound):

\begin{conjecture}
  For every field~$\FF$ and for all $p=p(n)$, the minimum rank of a fooling-set matrix with random zero-nonzero pattern $R(n,{p})$ is $\Omega( n )$.
\end{conjecture}

The bound in Thereom~\ref{thm:main}\ref{thm:main:small-p} is similar to that in~\cite{Golovnev-Regev-Weinstein:minrk:2016}, but it is better by roughly a factor of~$\log n$ if~$p$ is (constant or) slowly decreasing, e.g., $p=1/\log n$.  (Their minrank definition gives a lower bound to fooling-set pattern minimum rank.)

\bigskip\noindent%
The next three sections hold the proofs for Theorem~\ref{thm:main}.
%

\paragraph*{Acknowledgments.}
The second author would like to thank Kaveh Khoshkhah for discussions on the subject.

\section{Proof of Theorem~\ref{thm:main}\ref{thm:main:tiny-p}}
It is quite easy to see (using, e.g., Tur\'an's theorem) that in the region $p=O(1/n)$, $R(n,p)$ contains a triangular submatrix with nonzero diagonal entries of order $\Omega(n)$, thus lower bounding the rank over any field.  Here, we prove the following stronger result, which also gives a lower bound (for arbitrary fields) for more slowly decreasing~$p$.

\begin{lemma}\label{lem:indep-no}
  For $p(n) = d(n)/n = o(1)$, if $d(n) > C$ for some constant~$C$, then zero-nonzero pattern $R(n,p)$ contains a triangular submatrix with nonzero diagonal entries of size
  \begin{equation*}
    \Omega\biggl( \frac{\ln d}{d} \cdot n \biggr).
  \end{equation*}
\end{lemma}

We prove the lemma by using the following theorem about the independence number of random graphs in the Erd\H{o}s-R\'enyi model.  Let $G_{n,q}$ denote the random graph with vertex set~$[n]$ where each edge is chosen independently with probability~$q$.

\begin{theorem}[Theorem~7.4 in \cite{Janson-Luczak-Rucinski:Book}]\label{thm:Frz-indp}%
  Let $\epsilon>0$ be a constant, $q=q(n)$, and define
  \begin{equation*}
    k_{\pm\epsilon}  :=  \Bigl\lfloor \frac{2}{q}( \ln(nq)-\ln\ln(nq)+1-ln 2\pm\epsilon )   \Bigr\rfloor.
  \end{equation*}

  There exists a constant $C_{\epsilon}$ such that for $C_{\epsilon}/n\leq q=q(n)\leq \ln^{-2}n$, a.a.s., the largest independent set in~$G_{n,q}$ has size between $k_{-\epsilon}$ and $k_{+\epsilon}$.
\end{theorem}

\begin{proof}[Proof of Lemma~\ref{lem:indep-no}]
  Construct a graph~$G$ with vertex set $[n]$ from the fooling-set pattern matrix $R(n,p)$ in the following way: There is an edge between vertices $k$ and $\ell$ with $k>\ell$, if and only if $M_{k,\ell} \neq 0$.  This gives a random graph $G=G_{n,m,\nfrac12}$ which is constructed by first drawing uniformly at random a graph from all graphs with vertex set~$[n]$ and exactly~$m$ edges, and then deleting each edge, independently, with probability $\nfrac12$.  Using standard results in random graph theory (e.g., Lemma~1.3 and Theorem~1.4 in~\cite{Frieze-Karonski:rndGBook:2015}), this random graph behaves similarly to the Erd\H{o}s-R\'enyi graph with $q := p/2$.  In particular, since $G_{n,p/2}$ has an independent set of size $\Omega(n)$, so does $G_{n,m,\nfrac12}$.

  It is easy to see that the independent sets in $G$ are just the lower-triangular principal submatrices of $R_{n,p}$.
\end{proof}

As already mentioned, Theorem~\ref{thm:main}\ref{thm:main:tiny-p} is completed by noting that for $p < C/n$, an easy application of Tur\'an's theorem (or ad-hoc methods) gives us an independent set of size $\Omega(n)$.

\section{Proof of Theorem~\ref{thm:main}\ref{thm:main:small-p}}
Let $\FF$ be a finite field with $F:=\abs{\FF}$.  As mentioned in Theorem~\ref{thm:main}, we allow $F=F(n)$ to depend on~$n$.  In this section, we need to bound some quantities away from others, and we do that generously.

Let us say that a \textit{tee shape} is a set $T = I\times[n]\cup[n]\times I$, for some $I\subset[n]$.  A \textit{tee matrix} is a tee shape~$T$ together with a mapping $N\colon T\to\FF$ which satisfies
\begin{equation}\label{eq:small-p:def-fool-tee}
  N_{k,k}=1 \text{ for all $k\in I$, and } \quad N_{k,\ell}N_{\ell,k}=0 \text{ for all $(k,\ell)\in I\times[n]$, $k\ne\ell$.}
\end{equation}
The \textit{order} of the tee shape/matrix is $\abs{I}$, and the \textit{rank} of the tee matrix is the rank of the matrix $N_{I\times I}$.

For a matrix~$M$ and a tee matrix~$N$ with tee shape~$T$, we say that \textit{$M$ contains the tee matrix~$N$,} if $M_T = N$.

\begin{lemma}\label{lem:small-p:uniq-det-by-T}
  Let~$M$ be a matrix with rank $s :=\rk M$, which contains a tee matrix~$N$ of rank $s$.  Then~$M$ is the only matrix of rank~$s$ which contains~$N$.
\end{lemma}
In other words, the entries outside of the tee shape are uniquely determined by the entries inside the tee shape.
\begin{proof}
  Let $T = I\times[n]\cup[n]\times I$ be the tee shape of a tee matrix~$N$ contained in~$M$.

  Since $N_{I\times I} = M_{I\times I}$ and $\rk N_{I\times I} = s = \rk M$, there is a row set $I_1 \subseteq I$ of size $s=\rk M$ and a column set $I_2 \subseteq I$ of size~$s$ such that $\rk M_{I_1\times I_2} = s$.  This implies that~$M$ is uniquely determined, among the matrices of rank~$s$, by $M_{T'}$ with $T' := I_1\times[n]\cup[n]\times I_2 \subseteq T$.  (Indeed, since the rows of $M_{I_1\times [n]}$ are linearly independent and span the row space of~$M$, every row in $M$ is a unique linear combination of the rows in $M_{I_1\times [n]}$; since the rows in $M_{I_1\times I_2}$ are linearly independent, this linear combination is uniquely determined by the rows of $M_{[n]\times I_2}$.)

  Hence, $M$ is the only matrix $M'$ with $\rk M'=s$ and $M'_{T'} = M_{T'}$.
  Trivially, then, $M$ is the only matrix $M'$ with $\rk M'=s$ and $M'_{T} = M_{T} = N$.
\end{proof}

\begin{lemma}\label{lem:small-p:count-tee-matrices}
  For $r \le n/5$ and $m \le 2r(n-r)/3$, there are at most
  \begin{equation*}
    O(1) \cdot \binom{n}{2r} \cdot \binom{2r(n-r)}{m} \cdot (2F)^m
  \end{equation*}
  matrices of rank at most~$r$ over~$\FF$ which contain a tee matrix of order $2r$ with at most~$m$ nonzeros.
\end{lemma}
\begin{proof}
  By the Lemma~\ref{lem:small-p:uniq-det-by-T}, the number of these matrices is upper bounded by the number of tee matrices (of all ranks) of order~$2r$ with at most~$k$ nonzeros.

  The tee shape is uniquely determined by the set~$I\subseteq [n]$.  Hence, the number of tee shapes of order $2r$ is
  \begin{equation}\label{eq:small-p:owinccfwino}\tag{$*$}
    \binom{n}{2r}.
  \end{equation}

  The number of ways to choose the support a tee matrix.   Suppose that the tee matrix has~$h$ nonzeros.  Due to~\eqref{eq:small-p:def-fool-tee}, $h$ nonzeros must be chosen from $\binom{2r}{2}+2r(n-2r) \le 2r(n-r)$ opposite pairs.  Since $h < 2r(n-r)/2$, we upper bound this by
  \begin{equation*}
    \binom{ 2r(n-r) }{ h }.
  \end{equation*}
  For each of the~$h$ opposite pairs, we have to pick one side, which gives a factor of $2^h$.  Finally, picking, a number in $\FF$ for each of the entries designated as nonzero gives a factor of $(F-1)^h$.

  For summing over $h=0,\dots,m$, first of all, remember that $\sum_{i=0}^{(1-\eps)j/2}\binom{j}{i} = O_\eps(1)\cdot\binom{j}{(1-\eps)j/2}$ (e.g., Theorem 1.1 in~\cite{BollobasBkRndGraphs}, with $p=\nfrac12$, $u:=1+\eps$).  Since $m \le 2r(n-r)/3$, we conclude
  \begin{equation*}
    \sum_{h=0}^m \binom{ 2r(n-r) }{ h } = O(1)\cdot \binom{2r(n-r)}{m}
  \end{equation*}
  (with an absolute constant in the big-Oh).  Hence, we find that the number of tee matrices (with fixed tee shape) is at most
  \begin{equation*}
    \sum_{h=0}^m \binom{ 2r(n-r) }{ h } 2^h (F-1)^h
    \le (2F)^m \sum_{h=0}^m \binom{ 2r(n-r) }{ h } = O(1)\cdot (2F)^m \cdot \binom{2r(n-r)}{m}.
  \end{equation*}

  Multiplying by~\eqref{eq:small-p:owinccfwino}, the statement of the lemma follows.
\end{proof}

\begin{lemma}\label{lem:small-p:ex-tee-matrix}
  Let $r\le n/5$.  Every matrix~$M$ of rank at most~$r$ contains a tee matrix of order $2r$ and rank $\rk M$.
\end{lemma}
\begin{proof}
  There is a row set $I_1$ of size $s:=\rk M$ and a column set $I_2$ of size~$s$ such that $\rk M_{I_1\times I_2} = s$.  Take~$I$ be an arbitrary set of size $2r$ containing $I_1\cup I_2$, and $T:= I\times[n]\cup[n]\times I$.  Clearly, $M$ contains the tee matrix $N := M_T$, which is of order $2r$ and rank $s=\rk M$.
\end{proof}

\begin{lemma}\label{lem:small-p:no-dense-tee}
  Let $100\max(1,\ln\ln F)/n \le p \le 1$, and $n/(1000(\max(1,\ln F)) \le r \le n/100$.  A.a.s., every tee shape of order $2r$ contained in the random matrix $R(n,p)$ has fewer than $15pr(n-r)$ nonzeros.
\end{lemma}
\begin{proof}
  We take the standard Chernoff-like bound for the hypergeometric distribution of the intersection of uniformly random $p\binom{n}{2}$-element subset (the diagonally opposite pairs of $R(n,p)$ which contain a 1-entry) of a $\binom{n}{2}$-element ground set (the total number of diagonally opposite pairs) with a fixed $2r(n-r)$-element subset (the opposite pairs in~$T$) of the ground set:\footnote{Specifically, we use Theorem 2.10 applied to (2.11) in~\cite{Janson-Luczak-Rucinski:Book}}  With $\lambda := p2r(n-r)$ (the expected size of the intersection), if $x\ge 7\lambda$, the probability that the intersection has at least~$x$ elements is at most $e^{-x}$.

  Hence, the probability that the support of a fixed tee shape of order $2r$ is greater than than $15pr(n-r) \ge 14pr(n-r)+r$ is at most
  \begin{equation*}
    e^{-14pr(n-r)}
    \le
    e^{-r\cdot 14\cdot 99 \cdot \max(1,\ln\ln F)}
    \le
    e^{-r \cdot 1000 \cdot \max(1,\ln\ln F))}
  \end{equation*}

  Since the number of tee shapes is
  \begin{equation*}
    \binom{n}{r}
    \le e^{r(1+\ln(n/r))}
    \le e^{r(11+\ln\max(1,\ln F))},
    \le e^{r(11+\max(1,\ln\ln F))}
  \end{equation*}
  we conclude that the probability that a dense tee shape exists in $R(n,p)$ is at most $e^{-\Omega(r)}$.
\end{proof}

We are now ready for the main proof.
\begin{proof}[Proof of Theorem~\ref{thm:main}\ref{thm:main:small-p}]
  Call a fooling-set matrix~$M$ \textit{regular,} if $M_{k,k}=1$ for all~$k$.  The minimum rank over a fooling-set pattern is always attained by a regular matrix (divide every row by the corresponding diagonal element).

  Consider the event that there is a regular matrix~$M$ over~$\FF$ with $\sigma(M)=R(n,p)$, and $\rk M \le r := n/(2000\ln F)$.  By Lemma~\ref{lem:small-p:ex-tee-matrix}, $M$ contains a tee matrix~$N$ of order $2r$ and rank $\rk M$.  If the size of the support of~$N$ is larger than $15pr(n-r)$, then we are in the situation of Lemma~\ref{lem:small-p:no-dense-tee}.

  Otherwise, $M$ is one of the
  \begin{equation*}
    O(1) \cdot \binom{n}{2r} \cdot \binom{2r(n-r)}{15pr(n-r)} \cdot (2F)^{15pr(n-r)}
  \end{equation*}
  matrices of Lemma~\ref{lem:small-p:count-tee-matrices}.

  Hence, the probability of said event is $o(1)$ (from Lemma~\ref{lem:small-p:no-dense-tee}) plus at most an $O(1)$ factor of the following (with $m := pn^2/2$ and $\rho := r/n$) a constant
  \begin{multline*}
    \frac%
    {\displaystyle   \binom{n}{2r} \cdot \binom{2r(n-r)}{15pr(n-r)} \cdot (2F)^{15pr(n-r)}   }%
    {\displaystyle     \binom{\binom{n}{2}}{p\binom{n}{2}}  2^{p\binom{n}{2}} 2^{-O(pn)}     }
    =
    \frac%
    {\displaystyle   \binom{n}{2r} \cdot \binom{2r(n-r)}{15pr(n-r)} \cdot (2F)^{15pr(n-r)}   }%
    {\displaystyle     \binom{n^2/2}{pn^2/2}  2^{pn^2/2-O(pn)}                               }
    \\
    =
    \frac%
    {\displaystyle   \binom{n}{2\rho n} \cdot \binom{4\rho(1-\rho) n^2/2}{30p\rho(1-\rho)n^2/2} \cdot (2F)^{30p\rho(1-\rho)n^2/2}   }%
    {\displaystyle     \binom{n^2/2}{pn^2/2}  2^{pn^2/2-O(pn)}                                     }
    \\
    =
    \frac%
    {\displaystyle   \binom{n}{2\rho n} \cdot \binom{4\rho(1-\rho) \; n^2/2}{30\rho(1-\rho)\; pn^2/2} \cdot (2F)^{30\rho(1-\rho)\; pn^2/2}   }%
    {\displaystyle     \binom{n^2/2}{pn^2/2}  2^{pn^2/2-O(pn)}                                                                                     }
    =: Q
  \end{multline*}
  Abbreviating $\alpha := 30\rho(1-\rho)<30\rho$, denoting $H(t) := -t\ln t-(1-t)\ln(1-t)$, and using
  \begin{equation}\label{eq:binom-entropy:oinwfe}
    \binom{a}{t a} = \Theta\Bigl(  (ta)^{-\nfrac12}  \Bigr) e^{H(t)a},  \text{ for $t\le\nfrac12$}
  \end{equation}
  (for $a$ large, ``$\le$'' holds instead of ``$=\Theta$''), we find (the $O(pn)$ exponent comes from replacing $\binom{n}{2}$ by $n^2/2$ in the denominator)
  \begin{align*}
    \frac%
    {\displaystyle \binom{n}{2\rho n} 2^{30\rho(1-\rho)\; pn^2/2}      }%
    {\displaystyle                    2^{pn^2/2-O(pn)}                 }
    &\le
    e^{H(1/2\rho) n - (\ln 2) (1-\alpha)pn^2/3 }
    \\
    &\le
    e^{H(1/2\rho) n - (\ln 2) (1-\alpha)pn^2/3 }
    \\
    &=
    e^{n\bigl(  H(1/2\rho) - (\ln 2) (1-\alpha)pn/3 \bigr)}
    \\
    &\le
    e^{n\bigl(  H(1/2\rho) - (\ln 2) 33(1-30\rho) \bigr)}
    \\
    &=o(1),
  \end{align*}
  as $pn/2 \ge 30$ and $1-\alpha > 1-30\rho$, and the expression in the parentheses is negative for all $\rho\in[0, \, 3/100]$.

  For the rest of the fraction~$Q$ above, using~\eqref{eq:binom-entropy:oinwfe} again, we simplify
  \begin{equation*}
    \frac%
    {\displaystyle   \binom{4\rho(1-\rho) \; n^2/2}{30\rho(1-\rho)\; pn^2/2}  F^{30\rho(1-\rho)\; pn^2/2} }%
    {\displaystyle     \binom{n^2/2}{pn^2/2}                                  }
    \le
    \frac%
    {\displaystyle   \binom{\alpha \, n^2/2}{\alpha\, pn^2/2}    F^{30\rho(1-\rho)\; pn^2/2} }%
    {\displaystyle     \binom{n^2/2}{pn^2/2}                                                 }
    =
    O(1) \cdot
    e^{n^2/2 \cdot\bigl( (\alpha-1)H(p) + p\alpha\ln F \bigr)}.
  \end{equation*}
  Setting the expression in the parentheses to~$0$ and solving for~$\rho$, we find
  \begin{equation*}
    \alpha \ge \frac{\ln(1/p)}{\ln(1/p) + \ln F}
  \end{equation*}
  suffices for $Q = o(1)$; as $\alpha \le \rho$, the same inequality with $\alpha$ replaced by~$\rho$ is sufficient.
  This completes the proof of the theorem.
\end{proof}

\section{Proof of Theorem~\ref{thm:main}\ref{thm:main:const-p}}\label{sec:proof-const-p}
In this section, following the idea of~\cite{Hall-Hogben-Martin-Shader:exp-min-rk:2010}, we apply a theorem of Ronyai, Babai, and Ganapathy~\cite{Ronyai-Babai-Ganapathy:zero:2001} on the maximum number of zero-patterns of polynomials, which we now describe.

Let $f=(f_j)_{j=1,\dots,h}$ be an $h$-tuple of polynomials in~$n$ variables $x=(x_1,x_2,\cdots,x_n)$ over an arbitrary field~$\FF$.  In line with the definitions above, for $u\in \FF^n$, the zero-nonzero pattern of $f$ at~$u$ is the vector $\sigma(f(u)) \in \{0,1\}^h$.  

\begin{theorem}[\cite{Ronyai-Babai-Ganapathy:zero:2001}]\label{thm:RoBaGa}
  If 
  $h\ge n$ and each $f_j$ has degree at most $d$
  then, for all~$m$, the set
  \begin{equation*}
    \bbabs{  \Bigl\{    y \in \{0,1\}^h   \Bigm|  \abs{y} \le m   \text{ and }   y = \sigma(f(u)) \text{ for some~$u\in \FF^n$} \Bigr\}  } \le \binom{n+md}{n}.
  \end{equation*}
  In other words, the number of zero-nonzero patterns with Hamming weight at most~$m$ is at most $\binom{n+md}{n}$.
\end{theorem}
As has been observed in~\cite{Hall-Hogben-Martin-Shader:exp-min-rk:2010}, this theorem is implicit in the proof of Theorem~1.1 of~\cite{Ronyai-Babai-Ganapathy:zero:2001} (for the sake of completeness, the proof is repeated in the appendix).  %
It has been used in the context of minimum rank problems before (e.g., \cite{Mallik-Shader:skew-rk:2016,Hall-Hogben-Martin-Shader:exp-min-rk:2010}), but our use requires slightly more work.

\mypar%
Given positive integers $r < n$, let us say that a \textit{G-pattern} is an $r\times n$ matrix whose entries are the symbels $0$, $1$, and $*$, with the following properties.
\begin{enumerate}[label=(\arabic*),nosep]
\item Every column contains at most one~1, and every column containing a~1 contains no~$*$s.
\item In every row, the leftmost entry different from~0 is a~1, and every row contains at most one~1.
\item Rows containing a~1 (i.e., not all-zero rows) have smaller row indices than rows containing no~1 (i.e., all-zero rows).  In other words, the all-zero rows are at the bottom of~$P$.
\end{enumerate}

We say that an $r\times n$ matrix~$Y$ has \textit{G-pattern}~$P$, if
$Y_{j,\ell} = 0$ if $P_{j,\ell}=0$, %
and $Y_{j,\ell} = 1$ if $P_{j,\ell}=1$.  %
There is no restriction on the $Y_{j,\ell}$ for which $P_{j,\ell}=*$.

``G'' stands for ``Gaussian elimination using row operations''.  We will need the following tree easy lemmas.

\begin{lemma}\label{lem:rowop-pattern}
  Any $r\times n$ matrix~$Y'$ can be transformed, by Gaussian elimination using only row operations, into a matrix~$Y$ which has some G-pattern.
\end{lemma}
\begin{proof}[Proof (sketch).]
  If $Y'$ has no nonzero entries, we are done.  Otherwise   start with the left-most column containing a nonzero entry, say $(j,\ell)$.  Scale row~$j$ that entry a~1, permute the row to the top, and add suitable multiples of it to the other rows to make every entry below the~1 vanish.

  If all columns $1,\dots,\ell$ have been treated such that column~$\ell$ has a unique~1 in row, say $j(\ell)$, consider the remaining matrix $\{j(\ell)+1,\dots r\}\times\{\ell+1,\dots,n\}$.  If every entry is a~0, we are down.  Otherwise, find the leftmost nonzero entry in the block; suppose it is in column $\ell'$ and row $j'$.  Scale row~$j'$ to make that entry a~1, permute row~$j'$ to $j(\ell)+1$, and add suitable multiples of it to all other rows $\{1,\dots,r\}\setminus\{j(\ell)+1\}$ to make every entry below the~1 vanish.
\end{proof}

\begin{lemma}\label{lem:max-numo-ast-in-pattern}
  For every $r\times n$ G-pattern matrix~$P$, the number of $*$-entries in~$P$ is at most $r(n-r/2)$.
\end{lemma}
\begin{proof}[Proof (sketch).]
  The G-pattern matrix~$P$ is uniquely determined by $c_1<\dots<c_s$, the (sorted) list of columns of~$P$ which contain a~$1$.  With $c_0 := 0$, for $i=1,\dots,s$, if $c_{i-1} < c_i - 1$, then replacing $c_i$ by $c_i - 1$ gives us a G-pattern matrix with one more~$*$ entry.  Hence, we may assume that $c_i = i$ for $i=1,\dots,s$.  If $s<r$, then adding $s+1$ to the set of 1-columns cannot decrease the number of $*$-entries (in fact, it increases the number, unless $s+1=n$).  Hence, we may assume that $s=r$.  The number of $*$-entries in the resulting (unique) G-pattern matrix is
  \begin{equation*}
    n-1 + \dots + n-r = rn - r(r+1)/2 \le r(n - r/2),
  \end{equation*}
  as promised.
\end{proof}

\begin{lemma}\label{lem:numo-patterns}
  Let $\rho \in \lt]0,.49\rt]$.  The number of $n\times \rho n$ G-pattern matrices is at most
  \begin{equation*}
    O(1) \cdot \binom{n}{\rho n}
  \end{equation*}
  (with an absolute constant in the big-O).
\end{lemma}
\begin{proof}[Proof (sketch).]
  A G-pattern matrix is uniquely determined by the set of columns containing a~1, which can be between $0$ and $\rho n$.  Hence, the number of $n\times \rho n$ G-pattern matrices is
  \begin{equation}\label{eq:tmp:binomsum}\tag{$*$}
    \sum_{j=0}^{\rho n} \binom{ n }{ j }.
  \end{equation}
  From here on, we do the usual tricks.  As in the previous section, we use the helpful fact (Theorem 1.1 in~\cite{BollobasBkRndGraphs}) that
  \begin{equation*}
    \eqref{eq:tmp:binomsum}
    \le \frac{1}{1-\frac{\rho}{1-\rho}} \binom{n}{\rho n}.
  \end{equation*}
  A swift calculation shows that $1/(1-\rho/(1-\rho)) \le 30$, which completes the proof.
\end{proof}

We are now ready to complete the Proof of Theorem~\ref{thm:main}\ref{thm:main:const-p}.

\begin{proof}[Proof of Theorem \ref{thm:main}\ref{thm:main:const-p}]
  Let $M$ be a fooling-set matrix of size~$n$ and rank at most~$r$.  It can be factored as $M=XY$, for an $n\times r$ matrix~$X$ and an $r\times n$ matrix~$Y$.  By Lemma~\ref{lem:rowop-pattern}, through applying row operations to~$Y$ and corresponding column operations to~$X$, we can assume that~$Y$ has a G-pattern.

  Now we use Theorem~\ref{thm:RoBaGa}, for every G-pattern matrix separately.  For a fixed G-pattern matrix~$P$, the variables of the polynomials are
  \begin{itemize}
  \item $X_{k,j}$, where $(k,j)$ ranges over all pairs $\{1,\dots,n\} \times \{1,\dots,r\}$; and
  \item $Y_{j,\ell}$, where $(j,\ell)$ ranges over all pairs $\{1,\dots,r\} \times \{1,\dots,n\}$ with $P_{j,\ell}=*$.
  \end{itemize}
  The polynomials are: for every $(k,\ell)\in\{1,\dots,n\}^2$, with $k\ne \ell$,
  \begin{equation*}
    f_{k,\ell} = \sum_{\substack{j\\P_{j,\ell}=1}} X_{k,j} + \sum_{\substack{j\\P_{j,\ell}=*}} X_{k,j} Y_{j,\ell}.
  \end{equation*}
  Clearly, there are $n(n-1)$ polynomials; the number of variables is $2rn -r^2/2$, by Lemma~\ref{lem:max-numo-ast-in-pattern} (and, if necessary, using ``dummy'' variables which have coefficient~0 always).  The polynomials have degree at most~2.

  By Theorem~\ref{thm:RoBaGa}, we find that the number of zero-nonzero patterns with Hamming weight at most~$m$ of fooling-set matrices with rank at most~$r$ which result from this particular G-pattern matrix~$P$ is at most
  \begin{equation*}
    \binom{2rn -r^2/2 + 2m}{2rn -r^2/2}.
  \end{equation*}

  Now, take a $\rho < \nfrac12$, and let $r := \rho n$.  Summing over all G-pattern matrices~$P$, and using Lemma~\ref{lem:numo-patterns}, we find that the number of zero-nonzero patterns with Hamming weight at most~$m$ of fooling-set matrices with rank at most~$\rho n$ is at most an absolute constant times
  \begin{equation*}
    \binom{n}{\rho n}
    \binom{(2\rho -\rho^2/2)n^2 + 2m}{(2\rho -\rho^2/2)n^2}.
  \end{equation*}

  \myparbig%
  Now, take a constant ${p} \in \lt]0,1\rt]$, and let $m := \lceil {p}\binom{n}{2} \rceil$.  The number of fooling-set patterns of size~$n$ with density~${p}$ is
  \begin{equation*}
    \binom{ \binom{n}{2} }{ m } 2^m,
  \end{equation*}
  and hence, the probability that the minimum rank of a fooling-set matrix with zero-nonzero pattern $R(n,{p})$ has rank at most~$r$ is at most
  \begin{equation*}
    \frac%
    {\displaystyle \binom{n}{\rho n}\binom{(2\rho -\rho^2/2)n^2 + 2m}{(2\rho -\rho^2/2)n^2} }%
    {\displaystyle     \binom{ \binom{n}{2} }{ m } 2^m                                      }
    \le
    \frac%
    {\displaystyle \binom{n}{\rho n}\binom{(2\rho -\rho^2/2)n^2 + 2pn^2/2}{(2\rho -\rho^2/2)n^2} }%
    {\displaystyle     \binom{ n^2/2 }{ pn^2/2 } 2^{pn^2/2 +O(pn)}                               }
    =
    \frac%
    {\displaystyle \binom{n}{\rho n}\binom{\alpha n^2 + pn^2}{\alpha n^2} }%
    {\displaystyle     \binom{ n^2/2 }{ pn^2/2 } 2^{pn^2/2 +O(pn)}                               }
  \end{equation*}
  where we have set $\alpha := 2\rho -\rho^2/2$.
  As in the previous section, we use~\eqref{eq:binom-entropy:oinwfe} to estimate this expression, and we obtain
  \begin{equation*}
    \ln\lt(
    \frac%
    {\displaystyle \binom{n}{\rho n}\binom{\alpha n^2 + pn^2}{\alpha n^2} }%
    {\displaystyle     \binom{ n^2/2 }{ pn^2/2 } 2^{pn^2/2 +O(pn)}                               }
    \rt)
    =
    n H(\rho) + n^2\Bigr( \alpha H\bigl(\alpha/(\alpha+p)\bigr) - \tfrac12 H(p) -(\ln 2)p/2 \Bigr) + O(pn).
  \end{equation*}
  The dominant term is the one where~$n$ appears quadratic.  The expression $\tfrac12 H(p) +(\ln 2)p/2$ takes values in $\lt]0,1\rt[$.  For every fixed~$p$, the function $g\colon \alpha \mapsto \alpha H\bigl(\alpha/(\alpha+p)\bigr)$ is strictly increasing on $[0,\nfrac12]$ and satisfies $g(0)=0$.  Hence, for every given constant~$p$, there exists an $\alpha$ for which the coefficient after the $n^2$ is negative.

  (As indicated in the introduction, such an $\alpha$ must tend to~$0$ with $p\to 0$.)
\end{proof}


\input{main.bbl}
\appendix
\section{Proof of Lemma~\ref{lem:dietzfelbinger-minrk}}
Let $(x_1,y_1),\dots,(x_n,y_n) \in X\times Y$ be a fooling set in~$f$, and let~$A$ be a matrix over~$\FF$ with $A_{x,y} = 0$ iff $f(x,y)=0$.  Consider the matrix $B := A\otimes A^\Tp$.  This matrix~$B$ contains a permutation matrix of size~$n$ as a submatrix: for $i=1,\dots,n$, $B_{(x_i,x_i),(y_i,y_i)} = A_{x_i,y_i} A_{y_i,x_i} = 1$ but for $i\ne j$, $B_{(x_i,x_i),(y_j,y_j)} = A_{x_i,y_j} A_{y_i,x_j} = 0$.  Hence,
\begin{equation*}
  n \le \rk(B) = \rk(A)^2.
\end{equation*}
\qed

\section{Proof of Theorem~\ref{thm:RoBaGa}}
Since Theorem~\ref{thm:RoBaGa} is not explicitly proven in~\cite{Ronyai-Babai-Ganapathy:zero:2001}, we give here the slight modification of the proof of Theorem~1.1 from Theorem~\ref{thm:RoBaGa} which proves Theorem~\ref{thm:RoBaGa}.  The only difference between the following proof and that in~\cite{Ronyai-Babai-Ganapathy:zero:2001} is where the proof below upper-bounds the degrees of the polynomials $g_y$.
\begin{proof}[Proof of Theorem~\ref{thm:RoBaGa}]
  Consider the set
  \begin{equation*}
    S := \Bigl\{    y \in \{0,1\}^h   \Bigm|  \abs{y} \le m   \text{ and }   y = \sigma(f(u)) \text{ for some~$u\in \FF^n$} \Bigr\}.
  \end{equation*}
  For each such~$y$, let $u_y \in \FF^n$ be such that $\sigma(f(u_y)) = y$, and let
  \begin{equation*}
    g_y := \prod_{j, y_j=1} f_j.
  \end{equation*}
  Now define a square matrix~$A$ whose row- and column set is~$S$, and whose $(y,z)$ entry is $g_y(u_z)$.  We have
  \begin{equation*}
    g_y(u_z) \ne 0 \iff z \ge y,
  \end{equation*}
  with entry-wise comparison, and ``$1>0$''.  Hence, if the rows and columns are arranged according to this partial ordering of~$S$, the matrix is upper triangular, with nonzero diagonal, so it has full rank, $\abs{S}$.  This implies that the $g_y$, $y\in S$, are linearly independent.

  Since each~$g_y$ has degree at most $\abs{y}\cdot d \le md$, and the space of polynomials in~$n$ variables with degree at most~$md$ has dimension $\binom{n+md}{md}$, it follows that $S$ has at most that many elements.
\end{proof}
\end{document}

%% file: defsN.amsart.tex
\theoremstyle{plain}
\newtheorem{theorem}{Theorem}
\newtheorem*{theorem*}{Theorem}

\newtheorem*{proposition*}{Proposition}

\newtheorem*{corollary*}{Corollary}

\newtheorem{lemma}[theorem]{Lemma}
\newtheorem*{lemma*}{Lemma}

\newtheorem*{observation*}{Observation}

\newtheorem*{conjecture*}{Conjecture}
\newtheorem{conjecture}[theorem]{Conjecture}

\newtheorem*{question*}{Question}

\newtheorem*{questions*}{Questions}

\newtheorem*{problem*}{Problem}

\newtheorem*{problems*}{Problems}

\newtheorem*{openproblem*}{Open Problem}

\theoremstyle{definition}

\newtheorem*{definition*}{Definition}

\newtheorem*{example*}{Example}

\newtheorem*{exercise*}{Exercise}

\theoremstyle{remark}

\newtheorem*{remark*}{Remark}

\newtheorem*{remarks*}{Remarks}

\newtheorem*{claim*}{Claim}

\newcommand{\subclass}[1]{}

\newcommand{\enumTi}[1]{\renewcommand{\theenumi}{#1}}

\newcommand{\alphenumi}{\enumTi{\alph{enumi}}}

\newcommand{\romenumi}{\enumTi{\roman{enumi}}}


%% file: defs.common.tex
\newlength{\hspaceforlengthglumpf}

\renewcommand{\em}{\sl}

\DeclareMathOperator{\rk}{rk}

\newcommand{\lt}{\left}
\newcommand{\rt}{\right}

\newcommand{\Tp}{{\mspace{-1mu}\scriptscriptstyle\top\mspace{-1mu}}}

\newcommand{\abs}[1]{{\lt\lvert{#1}\rt\rvert}}

\newcommand{\bbabs}[1]{{\biggl\lvert{#1}\biggr\rvert}}

\newcommand{\nfrac}[2]{{\nicefrac{#1}{#2}}}

\newcommand{\FF}{\mathbb{F}}

\newcommand{\RR}{\mathbb{R}}

\newcommand{\eps}{\varepsilon}

\newlength{\algotabbingwidth}
